\documentclass[12pt]{article}
\usepackage[margin=25mm]{geometry}
\usepackage{amsmath,amssymb,amsthm}
\usepackage{hyperref}
\usepackage{enumerate}
\usepackage{algorithm}
\usepackage{algorithmic}

\newtheorem{theorem}{Theorem}[section]

\newtheorem{definition}[theorem]{Definition}
\newtheorem{example}[theorem]{Example}
\newtheorem{remark}[theorem]{Remark}
\newtheorem{question}[theorem]{Question}

\newcommand{\Eq}{E_q}
\newcommand{\Eqn}{E_q^n}
\renewcommand{\geq}{\geqslant}
\renewcommand{\leq}{\leqslant}

\title{The Combinatorial Rank of Subsets: Metric Density in Finite Hamming Spaces}

\author{Jamolidin K. Abdurakhmanov\\
Department of Information Technologies\\
Andijan State University\\
Andijan, Uzbekistan\\
\texttt{jamolidinkamol@gmail.com, abduraxmanov@adu.uz}\\
ORCID: 0009-0008-0067-0553}

\date{}

\begin{document}
\maketitle

\begin{abstract}
We introduce a novel concept of rank for subsets of finite metric spaces $\Eqn$ (the set of all $n$-dimensional vectors over an alphabet of size $q$) equipped with the Hamming distance, where the rank $R(A)$ of a subset $A$ is defined as the number of non-constant columns in the matrix formed by the vectors of $A$. This purely combinatorial definition provides a new perspective on the structure of finite metric spaces, distinct from traditional linear-algebraic notions of rank. We establish tight bounds for $R(A)$ in terms of $D_A$, the sum of Hamming distances between all pairs of elements in $A$. Specifically, we prove that $\frac{2qD_A}{(q-1)|A|^2} \leq R(A) \leq \frac{D_A}{|A|-1}$ when $|A|/q \geq 1$, with a modified lower bound for the case $|A|/q < 1$. These bounds show that the rank is constrained by the metric properties of the subset. Furthermore, we introduce the concept of metrically dense subsets, which are subsets that minimize rank among all isometric images. This notion captures an extremal property of subsets that represent their distance structure in the most compact way possible. We prove that subsets with uniform column distribution are metrically dense, and as a special case, establish that when $q$ is a prime power, every linear subspace of $\Eqn$ is metrically dense. This reveals a fundamental connection between the algebraic and metric structures of these spaces.
\end{abstract}

\section{Introduction}
\label{sec:intro}

The study of finite metric spaces equipped with the Hamming distance has been central to coding theory and combinatorics since the seminal work of Hamming \cite{hamming1950error}. The space $\Eqn$, consisting of all $n$-dimensional vectors over an alphabet of size $q$, forms the natural ambient space for error-correcting codes and has been extensively studied from various perspectives \cite{macwilliams1977theory,vanLint1999introduction}.

In this paper, we introduce a new notion of rank for subsets of $\Eqn$ that differs fundamentally from existing concepts in the literature. While traditional approaches focus on linear-algebraic rank or rank-metric codes \cite{gabidulin1985theory}, our definition is purely combinatorial: the rank $R(A)$ of a subset $A \subseteq \Eqn$ is the number of columns that vary across the vectors in $A$.

This seemingly simple definition leads to surprisingly rich mathematics. We establish optimal bounds for $R(A)$ in terms of the sum of pairwise Hamming distances within $A$, and introduce the concept of \emph{metrically dense} subsets—those that achieve minimal rank among all isometric images. Our main theoretical contribution shows that when $q$ is a prime power, every linear subspace is metrically dense, revealing an unexpected connection between algebraic and metric properties.

\subsection{Main Contributions}

Our work makes four primary contributions:

\begin{enumerate}
\item We introduce the rank $R(A)$ for subsets $A \subseteq \Eqn$, defined as the number of non-constant columns when the elements of $A$ are arranged as rows of a matrix.

\item We prove tight bounds for $R(A)$ based on the sum of Hamming distances $D_A = \sum_{x,y \in A, x \neq y} d_H(x,y)$:
\begin{itemize}
\item If $|A|/q \geq 1$: $\frac{2qD_A}{(q-1)|A|^2} \leq R(A) \leq \frac{D_A}{|A|-1}$
\item If $|A|/q < 1$ and $q > 2$: $\frac{2(q-2)D_A}{(|A|^2-2)(q-2)-(|A|-2)^2} \leq R(A) \leq \frac{D_A}{|A|-1}$
\end{itemize}

\item We introduce metrically dense subsets as those satisfying $R(A) \leq R(B)$ for all $B$ isometric to $A$.

\item We prove that when $q$ is a prime power, every linear subspace of $\Eqn$ is metrically dense, and more generally, any subset with uniform distribution in each varying column is metrically dense.
\end{enumerate}

\subsection{Related Work}

The notion of rank in coding theory typically refers to the rank of matrices over finite fields, as studied extensively in rank-metric codes \cite{delsarte1978bilinear,gabidulin1985theory}. The Hamming distance has been fundamental since \cite{hamming1950error}, with classical bounds including the Plotkin bound \cite{plotkin1960binary} and sphere-packing bounds \cite{macwilliams1977theory}.

Recent work on extremal problems in Hamming spaces \cite{ahlswede1997contributions,frankl1995extremal} has focused on the maximum size of subsets with prescribed distance sets. Our approach differs by considering the metric properties of subsets through their rank, providing a new lens for understanding finite metric spaces.

The connection between linear codes and metric properties has been explored through weight enumerators and MacWilliams identities \cite{macwilliams1962theorem}. Our result that linear subspaces are metrically dense adds a new dimension to this relationship.

\subsection{Paper Organization}

Section \ref{sec:prelim} establishes notation and presents our key definitions. Section \ref{sec:bounds} develops the bounds for $R(A)$ through careful analysis of column contributions. Section \ref{sec:dense} introduces metrically dense subsets and proves our main theorem about linear subspaces. Section \ref{sec:conclude} discusses open problems and potential applications.

\section{Preliminaries and Definitions}
\label{sec:prelim}

For natural numbers $n, q \geq 2$, we define:
\[
\Eq = \{0, 1, 2, \ldots, q - 1\}, \quad \Eqn = \{(x_1, x_2, \ldots, x_n) \mid x_i \in \Eq, \, i = 1, 2, \ldots, n\}.
\]

The set $\Eqn$ forms a metric space under the Hamming distance $d_H$, where for $x = (x_1, \ldots, x_n)$ and $y = (y_1, \ldots, y_n)$ in $\Eqn$:
\[
d_H(x, y) = |\{i : x_i \neq y_i\}|.
\]

When $q$ is a prime power, $\Eq$ can be identified with a finite field, making $\Eqn$ an $n$-dimensional vector space over $\Eq$.

\subsection{Faces and Rank}

A \emph{$k$-dimensional face} of $\Eqn$ is the set of all vectors where exactly $k$ coordinates vary freely while the remaining $n-k$ coordinates are fixed. The number of $k$-dimensional faces is $\binom{n}{k}q^{n-k}$.

\begin{definition}[Rank of a Subset]
\label{def:rank}
For any subset $A \subseteq \Eqn$, the \emph{rank} of $A$, denoted $R(A)$, is the dimension of the smallest face containing $A$.
\end{definition}

Equivalently, $R(A)$ can be computed by forming the matrix $M_A$ whose rows are the vectors in $A$:
\[
M_A = \begin{pmatrix}
a_{11} & a_{12} & \cdots & a_{1n} \\
\vdots & \vdots & \cdots & \vdots \\
a_{m1} & a_{m2} & \cdots & a_{mn}
\end{pmatrix},
\]
where $m = |A|$ and each row is a vector in $A$. A column is \emph{constant} if all its entries are equal. Then $R(A)$ equals the number of non-constant columns in $M_A$.

\begin{remark}
When $\Eqn$ is a linear space, the rank $R(L)$ of a linear subspace $L$ differs from its dimension $\dim L$. For example, if $L$ is a $k$-dimensional subspace contained in a hyperplane $\{x : x_i = 0\}$, then $\dim L = k$ but $R(L) \leq n-1$.
\end{remark}

\subsection{Isometry and Distance Sum}

Two subsets $A, B \subseteq \Eqn$ are \emph{isometric} if there exists a bijection $\phi: A \to B$ preserving Hamming distance: $d_H(\phi(x), \phi(y)) = d_H(x, y)$ for all $x, y \in A$.

\begin{example}
Let $n = 4$, $q = 2$, and consider:
\begin{align}
A &= \{(0,0,0,0), (0,0,1,1), (0,1,0,1), (0,1,1,0)\}, \\
B &= \{(0,0,0,0), (0,0,1,1), (0,1,0,1), (1,0,0,1)\}.
\end{align}
Both sets have pairwise distances all equal to 2, making them isometric. However, $R(A) = 3$ while $R(B) = 4$, showing that rank is not preserved under isometry.
\end{example}

For any subset $A \subseteq \Eqn$ with $|A| = m \geq 2$, we define:
\[
D_A = \sum_{\substack{x,y \in A \\ x \neq y}} d_H(x, y),
\]
the sum of all pairwise Hamming distances in $A$.

\section{Bounds for Rank}
\label{sec:bounds}

Our main technical result establishes tight bounds for $R(A)$ in terms of $D_A$ and $|A|$.

\begin{theorem}[Main Bounds]
\label{thm:main}
For any subset $A \subseteq \Eqn$ with $m = |A| \geq 2$:
\begin{enumerate}
\item If $m/q \geq 1$ and $q \geq 2$:
\[
\frac{2qD_A}{(q-1)m^2} \leq R(A) \leq \frac{D_A}{m-1}.
\]
\item If $m/q < 1$ and $q > 2$:
\[
\frac{2(q-2)D_A}{(m^2-2)(q-2)-(m-2)^2} \leq R(A) \leq \frac{D_A}{m-1}.
\]
\end{enumerate}
\end{theorem}

\begin{proof}
Let $r = R(A)$ be the number of non-constant columns in $M_A$, labeled $\bar{u}_1, \ldots, \bar{u}_r$. Each column $\bar{u}_i$ contributes some amount $d_i$ to $D_A$, with:
\[
D_A = \sum_{i=1}^r d_i.
\]

For column $\bar{u}_i$, let $y_j$ denote the number of occurrences of symbol $j \in \Eq$. Then:
\[
\sum_{j=0}^{q-1} y_j = m.
\]

The contribution of column $\bar{u}_i$ to $D_A$ is:
\[
d_i = \sum_{0 \leq k < j \leq q-1} y_k y_j = \frac{1}{2}\left(m^2 - \sum_{j=0}^{q-1} y_j^2\right).
\]

Since $\bar{u}_i$ is non-constant, at least two $y_j$ values are positive. Without loss of generality, assume $y_0 \geq 1$ and $y_1 \geq 1$.

\textbf{Upper bound:} To minimize $d_i$, we maximize $\sum_{j=0}^{q-1} y_j^2$. Subject to $y_0 \geq 1, y_1 \geq 1$, this is maximized when $y_0 = m-1, y_1 = 1$, and $y_j = 0$ for $j \geq 2$. Thus:
\[
d_i \geq \frac{1}{2}(m^2 - ((m-1)^2 + 1)) = m - 1.
\]

Therefore $D_A \geq r(m-1)$, giving $R(A) \leq D_A/(m-1)$.

\textbf{Lower bound:} To maximize $d_i$, we minimize $\sum_{j=0}^{q-1} y_j^2$.

If $m \geq q$: The minimum occurs when $y_j = m/q$ for all $j$, giving:
\[
d_i \leq \frac{1}{2}\left(m^2 - q \cdot \frac{m^2}{q^2}\right) = \frac{(q-1)m^2}{2q}.
\]

If $m < q$: Using Lagrange multipliers, the minimum occurs when $y_0 = y_1 = 1$ and $y_j = (m-2)/(q-2)$ for $j \geq 2$, giving:
\[
d_i \leq \frac{(m^2-2)(q-2)-(m-2)^2}{2(q-2)}.
\]

The lower bounds follow from $D_A \leq r \cdot \max(d_i)$.
\end{proof}

The bounds in Theorem \ref{thm:main} are tight:

\begin{example}[Tightness of bounds]
\begin{enumerate}
\item For $A = \{(0,0,0,0), (0,1,1,1)\} \subseteq E_2^4$: We have $|A| = 2$, $R(A) = 3$, $D_A = 3$. Both bounds give $R(A) = 3$.

\item For $A = \{(0,0,0), (0,2,2)\} \subseteq E_3^3$: We have $|A| = 2$, $R(A) = 2$, $D_A = 2$. Both bounds give $R(A) = 2$.
\end{enumerate}
\end{example}

\section{Metrically Dense Subsets}
\label{sec:dense}

\begin{definition}[Metrically Dense]
\label{def:dense}
A subset $A \subseteq \Eqn$ is \emph{metrically dense} if $R(A) \leq R(B)$ for every subset $B \subseteq \Eqn$ isometric to $A$.
\end{definition}

In other words, metrically dense subsets achieve minimal rank among all their isometric images. This concept captures an extremal property: such subsets represent their distance structure in the most compact way possible.

\begin{definition}[Uniform Column Distribution]
\label{def:uniform}
A subset $A \subseteq \Eqn$ has \emph{uniform column distribution} if, for each non-constant column of the matrix $M_A$, each value $b \in \Eq$ appears exactly $|A|/q$ times (assuming $q$ divides $|A|$).
\end{definition}

\begin{theorem}[Uniform Distribution Implies Dense]
\label{thm:uniform}
If $A \subseteq \Eqn$ has uniform column distribution and $q$ divides $|A|$, then $A$ is metrically dense.
\end{theorem}

\begin{proof}
Let $m = |A|$ with $q \mid m$. For each non-constant column in $M_A$, each value $b \in \Eq$ appears exactly $m/q$ times. The contribution of such a column to $D_A$ is:
\[
\binom{q}{2} \cdot \left(\frac{m}{q}\right)^2 = \frac{(q-1)q}{2} \cdot \frac{m^2}{q^2} = \frac{(q-1)m^2}{2q}.
\]

Since there are $R(A)$ such columns:
\[
D_A = R(A) \cdot \frac{(q-1)m^2}{2q}.
\]

This gives:
\[
R(A) = \frac{2qD_A}{(q-1)m^2},
\]
which equals the lower bound in Theorem \ref{thm:main}. Since no isometric image can have smaller rank, $A$ is metrically dense.
\end{proof}

As a corollary, we obtain our main result about linear subspaces:

\begin{theorem}[Linear Subspaces are Dense]
\label{thm:linear}
When $q$ is a prime power, every linear subspace $L \subseteq \Eqn$ with $|L| > 1$ is metrically dense.
\end{theorem}

\begin{proof}
Let $L$ be a $k$-dimensional linear subspace, so $|L| = q^k$. Consider any column $j$ of $M_L$ containing at least one nonzero entry. For each $b \in \Eq$, define:
\[
L_b = \{x \in L : x_j = b\}.
\]

Since $L$ is a subspace, $L_0 = \{x \in L : x_j = 0\}$ is also a subspace of dimension $k-1$. For any $b \neq 0$, there exists $z \in L$ with $z_j \neq 0$. The map $y \mapsto y + (b/z_j)z$ establishes a bijection between $L_0$ and $L_b$. Therefore, each value $b \in \Eq$ appears exactly $q^{k-1} = |L|/q$ times in column $j$.

Thus $L$ has uniform column distribution, and by Theorem \ref{thm:uniform}, $L$ is metrically dense.
\end{proof}

\begin{remark}
Theorem \ref{thm:uniform} shows that the property of being metrically dense is not exclusive to linear subspaces. Any subset with uniform column distribution achieves the lower bound in Theorem \ref{thm:main} and is therefore metrically dense. This includes many non-linear structures, such as certain combinatorial designs and regular graphs embedded in $\Eqn$.
\end{remark}

\begin{question}
Does every metrically dense subset have uniform column distribution? Equivalently, if $A$ is metrically dense with $q \mid |A|$, must each value appear exactly $|A|/q$ times in each non-constant column of $M_A$?
\end{question}

\section{Conclusions and Open Problems}
\label{sec:conclude}

We have introduced a new notion of rank for subsets of finite metric spaces with Hamming distance, established optimal bounds relating rank to distance structure, and proved that subsets with uniform column distribution achieve minimal rank among isometric images.

Several questions remain open:

\begin{enumerate}
\item \textbf{Characterization:} Can we characterize all metrically dense subsets of $\Eqn$? Are they precisely those with uniform column distribution?

\item \textbf{Non-divisibility case:} When $q$ does not divide $|A|$, which subsets are metrically dense? Can we extend the notion of uniform distribution to this case?

\item \textbf{Computational complexity:} Given a subset $A \subseteq \Eqn$, what is the complexity of determining whether $A$ is metrically dense?

\item \textbf{Applications:} Can metrically dense subsets be used to construct optimal codes or combinatorial designs?

\item \textbf{Generalizations:} Can these results be extended to other finite metric spaces or distance measures?
\end{enumerate}

Our results suggest that the rank function $R(\cdot)$ captures fundamental metric properties of finite spaces, with potential applications in coding theory, distributed storage, and combinatorial optimization. The characterization of metrically dense subsets through uniform column distribution provides a new perspective on the interplay between combinatorial and metric structures in finite spaces.

\end{document}